\documentclass[5p]{elsarticle}

\usepackage[T1]{fontenc}
\usepackage{float}
\usepackage{amsmath}
\usepackage{amsthm}
\usepackage{amssymb}
\usepackage{graphicx}
\usepackage[utf8]{inputenc}
\usepackage{bm}
\usepackage{hyperref}

\newtheorem{theorem}{Theorem}

\newtheorem{property}{Property}

\begin{document}
	
	\begin{frontmatter}
		
		\title{Imaginarity of quantum channels: Refinement and Alternative}
		
		\author{Xiangyu Chen}
		\ead{23S012014@stu.hit.edu.cn}
		
		\author{Qiang Lei\corref{cor1}}
		\ead{leiqiang@hit.edu.cn}
		
		\cortext[cor1]{Corresponding author}
		
		\affiliation{organization={School of Mathematics},
			addressline={Harbin Institute of Technology}, 
			city={Harbin},
			postcode={150001}, 
			country={China}}
		
		\begin{abstract}
			In this paper, we introduce the framework for quantifying the imaginarity of quantum channels. Besides, an alternative framework is given together to simplify the process of verifying the condition. We present three imaginarity measures of quantum channels based on the robustness, the trace norm, and the entropy, respectively. Some properties are also given.
		\end{abstract}
		
		
		
	\end{frontmatter}
	
	\section{Introduction}
	\par 
	Imaginary numbers offer a broad perspective and provide additional methodologies for modern scientific research. Quantum resource theory facilitates the establishment of new connections across various research directions. In recent years, the concept of imaginarity resource theory \cite{Hickey2018}, which has emerged alongside the development of quantum resource theory, has garnered increasing attention. Several applications can illustrate its significance. For instance, imaginarity plays a critical role in state discrimination \cite{Wu2021}. The imaginarity operation proves to be particularly advantageous for distinguishing quantum channels without the need for an auxiliary system \cite{Wu2023}. Numerous imaginarity measures have been identified for quantum states, including those based on fidelity \cite{Wu2021}, the convex roof \cite{Chen2022}, geometric considerations \cite{Wu2023}, and some induced by entropies \cite{Xu2023,Chen2024}. Some of them establish relations due to their mathematical properties \cite{Chen2024,Xue2021}. Just as quantum states warrant attention, quantum channels also merit investigation. Recent studies have begun to explore the properties of quantum channels \cite{Mani2015,Dana2017,Leditzky2018,Yuan2019,Wang2019}. The resource theories of quantum channels were first established in \cite{Liu2019}, with several subsequent theories proposed, such as coherence in quantum channels \cite{Xu2019}, entanglement in quantum channels \cite{Zhou2022}, and more \cite{Xu2021,Luo2022}. Building upon these frameworks, additional measures for quantum channels have been introduced \cite{Jin2021,Fan2022,Ye2024,Fan2024}.
	\par 
	The concept of the imaginarity of quantum channels was first introduced in \cite{Qiang2022}, and the definitions of free channels and free superchannels have been revisited recently \cite{Zanoni2024}. This indicates that the concept is of significant interest and worth considering. The structure of this paper is as follows: Similar to the coherence of quantum channels, we refine the concept of imaginarity for quantum channels and present an alternative framework that facilitates the verification of whether a given quantifier can serve as a suitable measure. We discuss the properties of imaginarity measures of quantum channels. Additionally, we propose three measures of the imaginarity of quantum channels based on the robustness, the trace norm, and the entropy, respectively.
	
	\section{Framework for quantifying the imaginarity of quantum channels}
	\par 
	Let $H_A$ and $H_B$ be two Hilbert space with dimensions $|A|$ and $|B|$, and orthonormal basis $\left\lbrace |j\rangle  \right\rbrace_{j=0}^{|A|-1}=\left\lbrace |k\rangle  \right\rbrace_{k=0}^{|A|-1} $ and $\left\lbrace |\alpha\rangle  \right\rbrace_{\alpha=0}^{|B|-1}=\left\lbrace |\beta\rangle  \right\rbrace_{\beta=0}^{|B|-1} $. We assume $|a\rangle\langle a|=E_{a,a}$ without loss of generality, that is, it happens to be the a-th diagonal. Denote $\mathcal{D}_A$ and $\mathcal{D}_B$ be the set of all density operators on $H_A$ and $H_B$, and $\mathcal{C}_{AB}$ be the set of all channels from $\mathcal{D}_A$ to $\mathcal{D}_B$. A quantum channel $\phi\in\mathcal{C}_{AB}$ can be represented by Choi matrix, 
	\begin{eqnarray*}
		J_{\phi}
		&=&\sum_{j,k}|j\rangle\langle k|\otimes \phi(|j\rangle\langle k|)
		\\
		&=&\sum_{j,k,\alpha,\beta}\phi_{j,k,\alpha,\beta}|j\rangle\langle k|\otimes |\alpha\rangle\langle \beta|
	\end{eqnarray*}
	with $\phi_{j,k,\alpha,\beta}=\langle\alpha|\phi(|j\rangle\langle k|)|\beta\rangle$.
	\par 
	The Choi state of $\phi$ is $\frac{J_{\phi}}{|A|}\triangleq\chi_{\phi}$. We call channel $\phi$ a real channel if $\chi_{\phi}$ is a real matrix, and denote $\mathcal{RC}_{AB}$ the set of all the real channels. The free channel is defined as the real channel. This is exactly the free operation in imaginarity resource theory, thus it is appropriate to define the free channel by Choi state.
	\par 
	Let $\mathcal{SC}_{ABA'B'}$ be the sets of all superchannels $\Theta$ from $\mathcal{C}_{AB}$ to $\mathcal{C}_{A'B'}$. The superchannel $\Theta$ can be represented by Choi matrix \cite{Xu2019},
	\begin{eqnarray*}
		J_{\Theta}
		&=&
		\sum_{j,k,\alpha,\beta}|j\alpha\rangle\langle k\beta|\otimes\Theta(|j\alpha\rangle\langle k\beta|)
		\\
		&=&
		\sum_{j,k,\alpha,\beta}\Theta_{j,k,\alpha,\beta,j',k',\alpha',\beta'}|j\alpha j'\alpha\rangle\langle k\beta k'\beta'|,
	\end{eqnarray*}
	where $\Theta_{j,k,\alpha,\beta,j',k',\alpha',\beta'}=\langle j'\alpha'|\Theta(|j\alpha\rangle\langle k\beta|)|k'\beta'\rangle$ with $|j\alpha\rangle$ means the tensor of the basis from systems $A$ and $B$. 
	\par 
	Similar to the case of channel, superchannel $\Theta$ also has the expression of Kraus operators $\left\lbrace M_m\right\rbrace_m $ with $\chi_{\Theta(\phi)}=\sum_m M_m\chi_{\phi}M_m^\dagger\triangleq\tilde{\Theta}(\chi_{\phi})$, where 
	\begin{eqnarray*}
		M_m=\sum_{j,j',\alpha,\alpha'}M_{m,j',j,\alpha',\alpha}|j'\alpha\rangle\langle j\alpha|
	\end{eqnarray*}
	with $\Theta_{j,k,\alpha,\beta,j',k',\alpha',\beta'}=\sum_mM_{m,j',j,\alpha',\alpha}M_{m,k',k,\beta',\beta}$.
	\par 
	Here, $\tilde{\Theta}$ means the channel corresponding to the superchannel \cite{Xu2019,Chiribella2008}. We will not distinguish the signs, $\Theta$ or $\tilde{\Theta}$, since it can be judged by the context.
	\par 
	We call a superchannel $\Theta: \mathcal{C}_{AB}\longrightarrow\mathcal{C}_{A'B'}$ a real superchannels  if $\Theta(\phi)\in\mathcal{RC}_{A'B'}$ for any channels $\phi\in\mathcal{RC}_{AB}$. We denote the set of all the real superchannels by $\mathcal{RSC}_{ABA'B'}$. The free superchannel is defined as the real superchannel.
	\par
	In \cite{Qiang2022}, the author presented a framework for quantifying the imaginarity of quantum channels. They introduced the imaginarity measure $C(\phi)\triangleq C(\chi_{\phi})$ of quantum channels should satisfy the following conditions:
	\par\noindent
	$\bm{(C_1)}$. Faithfulness:  $C(\phi)\geqslant 0$, and $C(\phi)= 0$ if and only if $\phi\in\mathcal{RC}_{AB}$;
	\par\noindent
	$\bm{(C_2)}$. Monotonicity:  $C(\phi)\geqslant C(\Theta(\phi))$ for any real superchannel $\Theta\in\mathcal{RSC}_{ABA'B'}$.
	\par 
	Actually, strong monotonicity and convexity often have a great use and a proper quantum resource theory should include them. So we add these to the conditions above, and then the imaginarity measure $C$ of quantum channels becomes more familiar and complete.
	\par\noindent
	$\bm{(C_3)}$. Strong monotonicity: Let $\Theta\in\mathcal{RSC}_{ABA'B'}$ with $\Theta(\cdot)=\sum_m M_m\cdot M_m^\dagger$, then $C(\phi)\geqslant\sum_m p_m C(\phi_m)$ where $p_m=\text{Tr}(M_m\chi_{\phi}M_m^\dagger)$ and $\chi_{\phi_m}=M_m\chi_{\phi}M_m^\dagger/p_m$
	\par\noindent
	$\bm{(C_4)}$. Convexity: For any $\left\lbrace \phi_m \right\rbrace\subset\mathcal{C}_{AB}$ and any probability distribution $\left\lbrace p_m \right\rbrace $, $C(\sum_m p_m\phi_m)\leqslant\sum_m p_m C(\phi_m)$.
	\par 
	We consider that a proper imaginarity measure of quantum channels should satisfy the condition $\bm{(C_1)}$--$\bm{(C_4)}$. In this way, we can obtain some properties of the imaginarity measure $C$. 
	\begin{property}\label{Pr1}
		Monotonicity under composition with any real channels:
		\par 
		$(a)$. Left composition: $C(\psi\circ\phi)\leqslant C(\phi)$ for any $\psi\in\mathcal{RC}_{BD}$ and any $\phi\in\mathcal{C}_{AB}$.
		\par 
		$(b)$. Right composition: $C(\phi\circ\psi)\leqslant C(\phi)$ for any $\psi\in\mathcal{RC}_{AB}$ and any $\phi\in\mathcal{C}_{BD}$.
		\par 
		$(C)$. Right composition: $C(\phi\circ\psi)=C(\phi)$ for any invertible $\psi\in\mathcal{RC}_{AB}$ and any $\phi\in\mathcal{C}_{BD}$.
	\end{property}
	\begin{proof}
		We only prove $(a)$, because $(b)$ is analogous. $(c)$ can be obtained from $(a)$ and $(b)$.
		\par 
		For $\psi\in\mathcal{RC}_{AB}$ and $\phi\in\mathcal{C}_{BD}$, we consider a superchannel $\Theta_{\psi}$ with $\Theta_{\psi}=\psi\circ\phi$. When superchannel $\Theta_{\psi}$ acts on any real channels, it is impossible to generate the imaginarity resource since the channel $\psi$ is real. It can be seen that the superchannel $\Theta_{\psi}$ is a real superchannel from the definition. The result holds due to the condition $\bm{(C_2)}$.
	\end{proof}
	\begin{property}\label{Pr2}
		The combination of imaginarity measures is still an imaginarity measure, that is, for any $\left\lbrace C_j\right\rbrace $ and any probability distribution $\left\lbrace q_j \right\rbrace $ with $q_j\geqslant 0$,  $C\triangleq\sum_jq_jC_j$ is an imaginarity measure.
	\end{property}
	\begin{proof}
		It is obvious that $C(\phi)\geqslant 0$. Due to $q_j\geqslant 0$, $C(\phi)= 0$ is equivalent to every $C_j(\phi)=0$, and it indicate $\phi\in\mathcal{RC}$. For any real superchannel $\Theta$, we have $C(\phi)=\sum_jq_jC_j(\phi)\geqslant \sum_jq_jC_j(\Theta(\phi))=C(\Theta(\phi))$.
		\par 
		For $\Theta\in\mathcal{RSC}$ with $\Theta(\cdot)=\sum_m M_m\cdot M_m^\dagger$, the strong monotonicity holds:
		\begin{eqnarray*}
			&&C(\phi)
			\\
			&=&\sum_jq_jC_j(\phi)
			\\
			&\geqslant&\sum_jq_j\sum_mp_mC_j(\phi_m)
			\\
			&=&\sum_mp_m\sum_jq_jC_j(\phi_m)
			\\
			&=&\sum_mp_mC(\phi_m).
		\end{eqnarray*}
		\par 
		For any $\left\lbrace \phi_m \right\rbrace\subset\mathcal{C}_{AB}$ and any probability distribution $\left\lbrace p_m \right\rbrace $, the convexity holds:
		\begin{eqnarray*}
			&&C(\sum_mp_m\phi_m)
			\\
			&=&\sum_jq_jC_j(\sum_mp_m\phi_m)
			\\
			&\leqslant&\sum_jq_j\sum_mp_mC_j(\phi_m)
			\\
			&=&\sum_mp_m\sum_jq_jC_j(\phi_m)
			\\
			&=&\sum_mp_mC(\phi_m).
		\end{eqnarray*}
		\par 
		Therefore, $C$ is an imaginarity measure.
	\end{proof}
	\begin{property}\label{Pr3}
		If $C$ is an imaginarity measure for channel $\phi$, then the quantifier $\hat{C}$ as
		\begin{eqnarray*}
			\hat{C}(\phi)\triangleq\min_{\left\lbrace q_j, \phi_j\right\rbrace }\sum_jq_jC(\phi_j),
		\end{eqnarray*} 
		is a proper measure for $\phi$, where $\min_{\left\lbrace q_j, \phi_j\right\rbrace }$ represents the min of all possible channels decompositions as $\phi=\sum_jq_j\phi_j$ with probability distribution $q_j\geqslant 0$.
	\end{property}
	\begin{proof}
		It can be seen that $\hat{C}(\phi)\geqslant 0$. When $\hat{C}(\phi)=0$, we can take  the best decomposition of $\phi$, expressed as $\phi=\sum_kq_k\phi_k$, so that $\hat{C}(\phi)=\sum_kq_kC(\phi_k)$. We can get $C(\phi_k)=0$ since $q\geqslant 0$, that is, $\phi_k\in\mathcal{RC}$, thus $\phi\in\mathcal{RC}$, the faithfulness holds. 
		\par 
		For any probability distribution $\left\lbrace q_j\right\rbrace $, if we suppose every $\phi_j=\phi$, we can have $\sum_jq_jC(\phi_j)= C(\phi)$. Thus, it follows that $\hat{C}(\phi)=\min_{\left\lbrace q_j, \phi_j\right\rbrace }\sum_jq_jC(\phi_j)\leqslant C(\phi)$. For any quantum channel $\phi$, we take the best decomposition of $\phi$, expressed as $\phi=\sum_kq_k\phi_k$, so that $\hat{C}(\phi)=\sum_kq_kC(\phi_k)$. For $\Theta\in\mathcal{RSC}$ with $\Theta(\cdot)=\sum_m M_m\cdot M_m^\dagger$, we denote $a_{k,m}=\text{Tr}(M_m\chi_{\phi_k}M_m^{\dagger})$, $p_m=\text{Tr}(M_m\chi_{\phi}M_m^{\dagger})$, then 
		\begin{eqnarray*}
			&&\hat{C}(\phi)
			\\
			&=&\sum_kq_kC(\phi_k)
			\\
			&\geqslant& \sum_kq_k\sum_ma_{k,m}C(\frac{M_m\chi_{\phi_k}M_m^{\dagger}}{a_{k,m}})
			\\
			&=&\sum_mp_m\sum_k\frac{q_ka_{k,m}}{p_m}C(\frac{M_m\chi_{\phi_k}M_m^{\dagger}}{a_{k,m}})
			\\
			&\geqslant&\sum_mp_mC(\sum_k\frac{q_ka_{k,m}}{p_m}\frac{M_m\chi_{\phi_k}M_m^{\dagger}}{a_{k,m}})
			\\
			&=&\sum_mp_mC(\frac{\sum_kq_kM_m\chi_{\phi_k} M_m^{\dagger}}{p_m})
			\\
			&=&\sum_mp_mC(\frac{M_m\chi_{\phi} M_m^{\dagger}}{p_m})
			\\
			&\geqslant&\sum_mp_m\hat{C}(\frac{M_m\chi_{\phi} M_m^{\dagger}}{p_m}).
		\end{eqnarray*}
		The first inequality follows from the strong monotonicity of $C$. The second inequality arises from the convexity of $C$ and the fact that $\sum_k\frac{q_ka_{k,m}}{p_m}=\frac{\sum_kq_k\text{Tr}(M_m\chi_{\phi_k}M_m^{\dagger})}{\text{Tr}(M_m\chi_{\phi}M_m^{\dagger})}=1$. The third inequality results from $\hat{C}(\phi)\leqslant C(\phi)$.
		\par 
		To verify the convexity of $\hat{C}$, for simplicity and without loss of generality, we consider only two quantum channels $\phi_1$ and $\phi_2$. We define $\psi=a\phi_1+b\phi_2$ with $a+b=1$. We take the best decomposition of $\phi_1$ and $\phi_2$ respectively as $\phi_1=\sum_mq_{1,m}\phi_{1,m}$ and $\phi_2=\sum_nq_{2,n}\phi_{2,n}$ so that the following equations hold:
		\begin{eqnarray*}
			\begin{cases}
				\hat{C}(\phi_1)=\sum_mq_{1,m}C(\phi_{1,m}),
				\\
				\hat{C}(\phi_2)=\sum_nq_{2,n}C(\phi_{2,n}).
			\end{cases}
		\end{eqnarray*}
		Then, $a\hat{C}(\phi_1)+b\hat{C}(\phi_2)=\sum_maq_{1,m}C(\phi_{1,m})+\sum_nbq_{2,n}C(\phi_{2,n})$.
		\par 
		For ease of writing, if $m\neq n$, we supplement the elements $\left\lbrace 0, \bm{0} \right\rbrace $ to the smaller set so that both sets have the same number of elements. This gives $a\hat{C}(\phi_1)+b\hat{C}(\phi_2)=\sum_maq_{1,m}C(\phi_{1,m})+bq_{2,m}C(\phi_{2,m})$. Next, we introduce the new symbols $\left\lbrace q_{2m-1}, \psi_{2m-1}\right\rbrace$ and $\left\lbrace q_{2m}, \psi_{2m}\right\rbrace$ to represent $\left\lbrace aq_{1,m}, \phi_{1,m}\right\rbrace $ and $\left\lbrace bq_{2,m}, \phi_{2,m}\right\rbrace $ respectively. Then we have 
		\begin{eqnarray*}
			&&a\hat{C}(\phi_1)+b\hat{C}(\phi_2)
			\\
			&=&\sum_mq_{2m-1}C(\psi_{2m-1})+q_{2m}C(\psi_{2m})
			\\
			&\geqslant&\min_{q_j, \psi_j}\sum_jq_jC(\psi_j)
			\\
			&=&\hat{C}(\psi)=\hat{C}(a\phi_1+b\phi_2).
		\end{eqnarray*}
		\par 
		In this way, we confirm that $\hat{C}$ is a valid imaginarity measure.
	\end{proof}
	
	\section{Alternative framework for the imaginarity of quantum channels}
	\par 
	\par 
	In Section 2, we introduce the framework. However, conditions $\bm{(C_3)}$ and $\bm{(C_4)}$ can sometimes be difficult to verify. In coherence resource theory, to simplify the process of verifying whether a quantifier qualifies as a coherence measure, an alternative framework for quantifying coherence has been proposed \cite{Yu2016}. Similar alternative frameworks can also be established in many other resource theories, such as imaginarity resource theory \cite{Xue2021}, and coherence of quantum channels \cite{Kong2022}. Inspired by \cite{Kong2022}, we Introduce the following condition:
	\par\noindent
	$\bm{(C_5)}$. Additivity: $C(\phi)=p_1C(\phi_1)+p_2C(\phi_2)$ for any $p_1$, $p_2$ satisfied $p_1+p_2=1$, any $\phi_1\in\mathcal{C}_{AB_1}$ and $\phi_2\in\mathcal{C}_{AB_2}$, where $\phi=p_1\phi_1\oplus p_2\phi_2$, $\phi\in\mathcal{C}_{AB}$ with $|B|=|B_1|+|B_2|$.
	\par 
	When conditions $\bm{(C_1)}$ and $\bm{(C_2)}$ hold, the conditions $\bm{(C_3)}$ and $\bm{(C_4)}$ are equivalent to the condition $\bm{(C_5)}$. To prove this, we first state the following theorem:
	\begin{theorem}\label{Th1}
		$\bm{(C_3)}$ and $\bm{(C_4)}$ give rise to $\bm{(C_5)}$ when condition $\bm{(C_1)}$ and $\bm{(C_2)}$ holds.
	\end{theorem}
	\begin{proof}
		\par\quad\par\noindent
		\textbf{STEP 1, obtain: $\bm{C(\phi)\geqslant p_1C(\tilde{\phi}_1)+p_2C(\tilde{\phi}_2)}$}.
		\par 
		Let $Q_1=|0\rangle\langle 0|+\cdots+||B_1|-1\rangle\langle |B_1|-1|+||B|\rangle\langle |B||+\cdots+||B|+|B_1|-1\rangle\langle |B|+|B_1|-1|+|(|A|-1)|B|\rangle\langle (|A|-1)|B||+\cdots+|(|A|-1)|B|+|B_1|-1\rangle\langle (|A|-1)|B|+|B_1|-1|$, 
		\par 
		$Q_2=||B_1|\rangle\langle |B_1||+\cdots+||B|-1\rangle\langle |B|-1|+||B|+|B_1|\rangle\langle |B|+|B_1||+\cdots+||B|+|B|-1\rangle\langle |B|+|B|-1|+|(|A|-1)|B|+|B_1|\rangle\langle (|A|-1)|B|+|B_1||+\cdots+|(|A|-1)|B|+|B|-1\rangle\langle (|A|-1)|B|+|B|-1|$. 
		\par 
		We can give an intuitive expression.
		\begin{eqnarray*}
			Q_1=
			\overbrace{\begin{pmatrix}
					I_{|B_1|} & & & & \\
					& O_{|B_2|} & & & \\
					& & \ddots & & \\
					& & & I_{|B_1|} & \\
					& & & & O_{|B_2|}
			\end{pmatrix}}^{|A||B|}
			\\
			=\bigoplus\limits_{j=0}^{|A|-1}
			\begin{pmatrix}
				I_{|B_1|} & \\
				& O_{|B_2|}
			\end{pmatrix},
		\end{eqnarray*}
		\begin{eqnarray*}
			Q_2=
			\underbrace{\begin{pmatrix}
					O_{|B_1|} & & & & \\
					& I_{|B_2|} & & & \\
					& & \ddots & & \\
					& & & O_{|B_1|} & \\
					& & & & I_{|B_2|}
			\end{pmatrix}}_{|A||B|}
			\\
			=\bigoplus\limits_{j=0}^{|A|-1}
			\begin{pmatrix}
				O_{|B_1|} & \\
				& I_{|B_2|}
			\end{pmatrix},
		\end{eqnarray*}
		where $I_a$ and $O_b$ denote identity matrix and zero matrix with dimensions $a$ and $b$ respectively.
		\par 
		Given that $Q_1^{\dagger}Q_1+Q_2^{\dagger}Q_2=I_{|A||B|}$ and $Q_1Q_2=O_{|A||B|}$, let $\Theta(\cdot)=Q_1\cdot Q_1^{\dagger}+Q_2\cdot Q_2^{\dagger}$. It follows that $\Theta\in\mathcal{RSC}_{ABAB}$.
		\par 
		Under the requirement of $\bm{(C_5)}$, the Choi state of $\phi$ is
		\begin{eqnarray*}
			&&\chi_{\phi}
			\\
			&=&\sum_{j,k}|j\rangle\langle k|\otimes \phi(|j\rangle\langle k|)/|A|
			\\
			&=&\sum_{j,k}|j\rangle\langle k|\otimes \big[ p_1\phi_1(|j\rangle\langle k|)\oplus p_2\phi_2(|j\rangle\langle k|\big] /|A|.
		\end{eqnarray*}
		Thus, 
		\begin{eqnarray*}
			Q_1\chi_{\phi}Q_1^{\dagger}=\sum_{j,k}|j\rangle\langle k|\otimes \big[ p_1\phi_1(|j\rangle\langle k|)\oplus \bm{0}_2(|j\rangle\langle k|\big] /|A|.
		\end{eqnarray*}
		It is clear that $\text{Tr}(Q_1\chi_{\phi}Q_1^{\dagger})$ is exactly $p_1$. Let $\chi_{\tilde{\phi}_1}=Q_1\chi_{\phi}Q_1^{\dagger}/p_1=\sum_{j,k}|j\rangle\langle k|\otimes \big[ \phi_1(|j\rangle\langle k|)\oplus \bm{0}_2(|j\rangle\langle k|\big] /|A|$. Similarly we can obtain $\chi_{\tilde{\phi}_2}$. From the condition $\bm{(C_3)}$, we have:
		\begin{eqnarray}\label{1}
			C(\phi)=C(\chi_{\phi})
			&\geqslant& p_1C(\chi_{\tilde{\phi}_1})+p_2C(\chi_{\tilde{\phi}_2}) \nonumber
			\\
			&=&p_1C(\tilde{\phi}_1)+p_2C(\tilde{\phi}_2).
		\end{eqnarray}
		\par\noindent
		\textbf{STEP 2, obtain: $\bm{C(\phi)\leqslant p_1C(\tilde{\phi}_1)+p_2C(\tilde{\phi}_2)}$}.
		\par 
		For any $\rho\in\mathcal{D}_A$, $\rho=\sum_{j,k}\rho_{j,k}|j\rangle\langle k|$, any $\phi\in\mathcal{C}_{AB}$, $\phi(\rho)$ can be represented by $J_{\phi}$, in order to do this, $\rho$ and $J_{\phi}$ need to be tensor index realigned: 
		\begin{eqnarray*}
			\begin{cases}
				J'_{\phi}=\sum_{j,k}\langle jk|\otimes \phi(|j\rangle\langle k|),
				\\
				\rho'=\sum_{j,k}\rho_{j,k}|jk\rangle,
			\end{cases}
		\end{eqnarray*}
		then $\phi(\rho)=J'_{\phi}\rho'=\sum_{j,k}\rho_{j,k}\phi(|j\rangle\langle k|)$. Under the requirement of $(C_5)$, for any $\rho$ we have
		\begin{eqnarray*}
			&&\phi(\rho)
			\\
			&=&\sum_{j,k}\rho_{j,k}\phi(|j\rangle\langle k|)
			\\
			&=&\sum_{j,k}\rho_{j,k}\big[p_1\phi_1(|j\rangle\langle k|)\oplus p_2\phi_2(|j\rangle\langle k|)\big]
			\\
			&=&p_1\sum_{j,k}\rho_{j,k}\big[\phi_1(|j\rangle\langle k|)\oplus \bm{0}_2(|j\rangle\langle k|)\big]
			\\
			&&+p_2\sum_{j,k}\rho_{j,k}\big[\bm{0}_1(|j\rangle\langle k|)\oplus \phi_2(|j\rangle\langle k|)\big]
			\\
			&=&p_1\tilde{\phi}_1(\rho)+p_2\tilde{\phi}_2(\rho),
		\end{eqnarray*}
		that is $\phi=p_1\tilde{\phi}_1+p_2\tilde{\phi}_2$. From $\bm{(C_4)}$, it follows that
		\begin{eqnarray}\label{2}
			C(\phi)\leqslant p_1C(\tilde{\phi}_1)+p_2C(\tilde{\phi}_2).
		\end{eqnarray}
		\par\noindent
		\textbf{STEP 3, obtain $\bm{(C_5)}$ since $\bm{C(\phi_j)=C(\tilde{\phi}_j)}$, $\bm{j=1,2}$}. 
		\par 
		From inequalities (\ref{1}) and (\ref{2}) above, we obtain $C(\phi)=p_1C(\tilde{\phi}_1)+p_2C(\tilde{\phi}_2)$. Next, we need to prove that $C(\phi_1)=C(\tilde{\phi}_1)$, which requires us to show that $C(\chi_{\phi_1})=C(\chi_{\tilde{\phi}_1})$. 
		\par 
		First, we specify the range of indices $i$ and $j$, let $i=0, 1, \cdots, |B|-1$, $j=0, 1, \cdots, |B_1|-1$. Let $K_0^1$ satisfies $\langle i|K_0^1|j\rangle=\delta_{i,j}$. An intuitive expression for $K_0^1$ is 
		\begin{eqnarray*}
			K_0^1=\left.\begin{pmatrix}
				I_{|B_1|} \\
				\hat{O}
			\end{pmatrix}\right\rbrace |B|, 
		\end{eqnarray*}
		where $\hat{O}$ denotes the zero matrix with dimension $|B_2|\times |B_1|$.
		\par 
		Let $K_n^2$ satisfy $\langle j|K_n^2|i\rangle=\delta_{i,j+n|B_1|}$ for $n=0, 1, \cdots,$ $ \left\lceil |B_2|/|B_1|\right\rceil$, where $\left\lceil a\right\rceil$ means the ceil of $a$.
		\par 
		Let
		\begin{eqnarray*}
			\mathcal{K}_0^1=\bigoplus_{j=0}^{|A|-1}K_0^1=\begin{pmatrix}
				K_0^1 & & \\
				& \ddots & \\
				& & K_0^1
			\end{pmatrix}, 
			\\
			\mathcal{K}_n^2=\bigoplus_{j=0}^{|A|-1}K_n^2=\begin{pmatrix}
				K_n^2 & & \\
				& \ddots & \\
				& & K_n^2
			\end{pmatrix}.
		\end{eqnarray*}
		\par 
		 We have $(\mathcal{K}_0^1)^{\dagger}\mathcal{K}_0^1=I_{|B_1|}$ and $\sum_{j=0}^{\left\lceil |B_2|/|B_1|\right\rceil}(\mathcal{K}_n^2)^{\dagger}\mathcal{K}_n^2=I_{|B|}$. Thus $\Theta_1\in\mathcal{RSC}_{AB_1AB}$ with $\Theta_1(\cdot)=\mathcal{K}_0^1\cdot(\mathcal{K}_0^1)^{\dagger}$,  $\Theta_2\in\mathcal{RSC}_{ABAB_1}$ with $\Theta_2(\cdot)=\sum_{j=0}^{\left\lceil |B_2|/|B_1|\right\rceil}\mathcal{K}_n^2\cdot (\mathcal{K}_n^2)^{\dagger}$. It can be seen that $\Theta_1(\chi_{\phi_1})=\chi_{\tilde{\phi}_1}$ and $\Theta_2(\chi_{\tilde{\phi}_1})=\chi_{\phi_1}$. 
		\par 
		Combining with $\bm{(C_2)}$, the above implies that the following inequality holds: 
		\begin{eqnarray*}
			&&C(\chi_{\phi_1})
			\\
			&=&C(\Theta_2(\chi_{\tilde{\phi}_1}))
			\\
			&\leqslant& C(\chi_{\tilde{\phi}_1})
			\\
			&=&C(\Theta_1(\chi_{\phi_1}))
			\\
			&\leqslant& C(\chi_{\phi_1}).
		\end{eqnarray*}
		As a result, $C(\chi_{\phi_1})=C(\chi_{\tilde{\phi}_1})$. It follows from a similar discussion that $C(\chi_{\phi_2})=C(\chi_{\tilde{\phi}_2})$.
		\par 
		Up to now, we prove $\bm{(C_5)}$ holds: $C(\phi)=p_1C(\phi_1)+p_2C(\phi_2)$.
	\end{proof}
	\par 
	Under the requirement of $\bm{(C_5)}$, we have $\chi_{\phi_1}$ and $\chi_{\phi_2}$ below.
	\begin{eqnarray*}
		\begin{cases}
			\chi_{\phi_1}=\sum_{j,k}|j\rangle\langle k|\otimes\phi_1(|j\rangle\langle k|)/|A|,
			\\
			\chi_{\phi_2}=\sum_{j,k}|j\rangle\langle k|\otimes\phi_2(|j\rangle\langle k|)/|A|.
		\end{cases}
	\end{eqnarray*}
	We also define $\chi_{\tilde{\phi}_1}$ and $\chi_{\tilde{\phi}_2}$ below. 
	\begin{eqnarray*}
		\begin{cases}
			\chi_{\tilde{\phi}_1}=\sum_{j,k}|j\rangle\langle k|\otimes \big[ \phi_1(|j\rangle\langle k|)\oplus \bm{0}_2(|j\rangle\langle k|\big] /|A|,
			\\
			\chi_{\tilde{\phi}_2}=\sum_{j,k}|j\rangle\langle k|\otimes \big[ \bm{0}_1(|j\rangle\langle k|)\oplus \phi_2(|j\rangle\langle k|\big] /|A|.
		\end{cases}
	\end{eqnarray*}
	Thus the following property holds.
	\par
	Next, we prove the reverse result:
	\begin{theorem}\label{Th2}
		$\bm{(C_5)}$ gives rise to $\bm{(C_3)}$ and $\bm{(C_4)}$ when condition $\bm{(C_1)}$ and $\bm{(C_2)}$ holds.
	\end{theorem}
	\begin{proof}
		\par\quad\par\noindent
		\textbf{prove $\bm{(C_3)}$}. 
		\par 
		For $\Theta\in\mathcal{RSC}_{ABA'B'}$ with  $\Theta(\cdot)=\sum_m M_m\cdot M_m^\dagger$, we introduce an auxiliary system $S$ with dimension $N$ and orthonormal basis $\left\lbrace |n\rangle\right\rbrace_{n=0}^{N-1} $. The system $AB$, formed by Choi state and the auxiliary system $S$, constitutes a combined system $ABS$, whose basis is $\left\lbrace |j\alpha\rangle\otimes|n\rangle\right\rbrace $. The entire system is initially in the state $\chi_{\phi}\otimes |0\rangle\langle 0|\triangleq\chi_{\phi}^{ABS}$. 
		\par 
		Let $U_n=\sum_{k=0}^{N-1}|(k+n) \text{ mod } N\rangle\langle k|$, then $\Theta^{ABS}(\cdot)\in\mathcal{RSC}_{ABSA'B'S}$ with $\Theta^{ABS}(\cdot)=\sum_n(M_n\otimes U_n)\cdot (M_n\otimes U_n)^{\dagger}$. Thus
		\begin{eqnarray*}
			&&\Theta^{ABS}(\chi_{\phi}^{ABS})
			\\
			&=&\sum_n(M_n\otimes U_n)(\chi_{\phi}\otimes |0\rangle\langle 0|)(M_n\otimes U_n)^{\dagger}
			\\
			&=&\sum_nM_n\chi_{\phi}M_n^{\dagger}\otimes |n\rangle\langle n|
			\\
			&=&\sum_np_n\chi_{\phi_n}\otimes |n\rangle\langle n|,
		\end{eqnarray*}
		where $p_n=\text{Tr}(M_n\chi_{\phi}M_n^{\dagger})$ and $\chi_{\phi_n}=M_n\chi_{\phi}M_n^{\dagger}/p_n$.
		\par 
		It can be seen that $\sum_np_n=1$, $\phi_n\in\mathcal{C}_{AB_n}$ with $|B|=\sum_n|B_n|$. From $\bm{(C_5)}$, we have
		\begin{eqnarray*}
			C(\chi_{\phi}^{ABS})=C(\chi_{\phi}\otimes |0\rangle\langle 0|)=C(\chi_{\phi}),
		\end{eqnarray*}
		and
		\begin{eqnarray*}
			&&C(\Theta^{ABS}(\chi_{\phi}^{ABS}))
			\\
			&=&C(\sum_np_n\chi_{\phi_n}\otimes |n\rangle\langle n|)
			\\
			&=&C(\bigoplus_np_n\chi_{\phi_n})
			\\
			&=&\sum_np_nC(\chi_{\phi_n}).
		\end{eqnarray*}
		From $\bm{(C_2)}$ and above, we can obtain
		\begin{eqnarray*}
			C(\chi_{\phi})\geqslant\sum_np_nC(\chi_{\phi_n}),
		\end{eqnarray*}
		that is, condition $\bm{(C_3)}$ holds.
		\par\noindent
		\textbf{prove $\bm{(C_4)}$}. 
		\par 
		For any $\left\lbrace \phi_m\right\rbrace\subset\mathcal{C}_{AB}$ and any probability distribution $\left\lbrace p_m \right\rbrace $, we construct a auxiliary system $S$ with orthonormal basis $\left\lbrace |m\rangle \right\rbrace_m$ and define $\chi_{\phi}^{ABS}=\sum_mp_m\chi_{\phi_m} \otimes|m\rangle\langle m|$. Notably, we have:
		\begin{eqnarray*}
			\sum_m(I_{|A||B|}\otimes|0\rangle\langle m|)^{\dagger}(I_{|A||B|}\otimes|0\rangle\langle m|)=I_{|A||B||S|}
		\end{eqnarray*}
		 and $\Theta^{ABS}\in\mathcal{RSC}_{ABSA'B'S}$ with $\Theta^{ABS}(\cdot)=\sum_m(I_{|A||B|}\otimes|0\rangle\langle m|)\cdot (I_{|A||B|}\otimes|0\rangle\langle m|)^{\dagger}$. Then,
		\begin{eqnarray*}
			&&\Theta^{ABS}(\chi_{\phi}^{ABS})
			\\
			&=&\sum_m(I_{|A||B|}\otimes|0\rangle\langle m|)\chi_{\phi}^{ABS} (I_{|A||B|}\otimes|0\rangle\langle m|)^{\dagger}
			\\
			&=&\sum_mp_m\chi_{\phi_m}\otimes|0\rangle\langle 0|.
		\end{eqnarray*}
		\par 
		From $\bm{(C_2)}$ and $\bm{(C_5)}$, the following inequality holds: 
		\begin{eqnarray*}
			&&\sum_mp_mC(\chi_{\phi_m})
			\\
			&=&C(\chi_{\phi}^{ABS})
			\\
			&\geqslant& C(\Theta^{ABS}(\chi_{\phi}^{ABS}))
			\\
			&=&C(\sum_mp_m\chi_{\phi_m}),
		\end{eqnarray*}
		that is, condition $\bm{(C_4)}$ holds.
	\end{proof}
	Therefore, from the above two theorems, we can give the Alternative framework for the imaginarity of quantum channels:
	\begin{theorem}\label{Th3}
		$\bm{(C_3)}+\bm{(C_4)}$ is equivalent to $\bm{(C_5)}$ when condition $\bm{(C_1)}$ and $\bm{(C_2)}$ holds. 
	\end{theorem}
	When verifying whether a quantifier is an imaginarity measure, we can choose to verify either the conditions $\bm{(C_1)}$--$\bm{(C_4)}$ or the conditions $\bm{(C_1)}+\bm{(C_2)}+\bm{(C_5)}$. 
	\begin{property}\label{Pr4}
		$C$ is a proper imaginarity measure for quantum channels, the equation holds:  
		\begin{eqnarray}\label{3}
			C(p_1\chi_{\phi_1}\oplus p_2\chi_{\phi_2})=C(p_1\chi_{\tilde{\phi}_1}+p_2\chi_{\tilde{\phi}_2})=C(\chi_{\phi}).
		\end{eqnarray}
	\end{property}
	\begin{proof}
		We employ the same idea as in Theorem \ref{Th1}, specifically using $\bm{(C_2)}$ to establish equivalence. Without loss of generality, we assume $|A|=3$, $p_1\phi_1(|j\rangle\langle k|)\triangleq X_{jk}$ and $p_2\phi_2(|j\rangle\langle k|)\triangleq Y_{jk}$. Then,
		\begin{eqnarray*}
			p_1\chi_{\phi_1}\oplus p_2\chi_{\phi_2}=
			\begin{pmatrix}
				X_{11} & X_{12} & X_{13} & & & \\
				X_{21} & X_{22} & X_{23}& & & \\
				X_{31} & X_{32} & a_{33} & & & \\
				& & & Y_{11} & Y_{12} & Y_{13} \\
				& & & Y_{21} & Y_{22} & Y_{23} \\
				& & & Y_{31} & Y_{32} & Y_{33}
			\end{pmatrix},
		\end{eqnarray*}
		and 
		\begin{eqnarray*}
			p_1\chi_{\tilde{\phi}_1}+p_2\chi_{\tilde{\phi}_2}=
			\begin{pmatrix}
				X_{11} & & X_{12} & & X_{13} & \\
				& Y_{11} & & Y_{12} & & Y_{13} \\
				X_{21} & & X_{22} & & X_{23} & \\
				& Y_{21} & & Y_{22} & & Y_{23} \\
				X_{31} & & X_{32} & & X_{33} & \\
				& Y_{31} & & Y_{32} & & Y_{33}
			\end{pmatrix}.
		\end{eqnarray*}
		In this way, we need to find some superchannels between $p_1\chi_{\phi_1}\oplus p_2\chi_{\phi_2}$ and $\chi_{\phi}$. We can set
		\begin{eqnarray*}
			K_1=
			\begin{pmatrix}
				I_{|B_1|} & & & & & \\
				& & & I_{|B_1|} & & \\
				& & & & I_{|B_1|} & \\
				& I_{|B_2|} & & & & \\
				& & I_{|B_2|} & & & \\
				& & & & & I_{|B_2|}
			\end{pmatrix}, 
			\\
			K_2=
			\begin{pmatrix}
				I_{|B_1|} & & & & & \\
				& I_{|B_2|} & & & & \\
				& & I_{|B_1|} & & & \\
				& & & I_{|B_2|} & & \\
				& & & & I_{|B_1|} & \\
				& & & & & I_{|B_2|}
			\end{pmatrix}. 
		\end{eqnarray*}
		It can be known that $K^{\dagger}_jK_j=I$, $j=1,2$, and the following transformations,
		\begin{eqnarray*}
			p_1\chi_{\phi_1}\oplus p_2\chi_{\phi_2}
			\xrightarrow{K_1(\cdot)K_1^{\dagger}}
			N
			\xrightarrow{K_2(\cdot)K_2^{\dagger}}
			p_1\chi_{\tilde{\phi}_1}+p_2\chi_{\tilde{\phi}_2},
		\end{eqnarray*}
		where $N$ is
		\begin{eqnarray*}
			\begin{pmatrix}
				X_{11} & & & X_{12} & X_{13} & \\
				& Y_{11} & Y_{12} & & & Y_{13} \\
				& Y_{21} & Y_{22} & & & Y_{23} \\
				X_{21} & & & X_{22} & X_{23} & \\
				X_{31} & & & X_{32} & X_{33} & \\
				& Y_{31} & Y_{32} & & & Y_{33}
			\end{pmatrix}. 
		\end{eqnarray*}
		We define $\Theta_j(\cdot)=K_j\cdot K_j^{\dagger}$, $j=1,2$. It is clear that $\Theta_j\in\mathcal{RSC}$. Then the inequality holds:
		\begin{eqnarray*}
			C(p_1\chi_{\phi_1}\oplus p_2\chi_{\phi_2})\geqslant C(p_1\chi_{\tilde{\phi}_1}+p_2\chi_{\tilde{\phi}_2}).
		\end{eqnarray*}
		Since $K_1$ and $K_2$ are invertible, the inverse inequality also holds. Therefore, we have:
		\begin{eqnarray*}
			C(p_1\chi_{\phi_1}\oplus p_2\chi_{\phi_2})= C(p_1\chi_{\tilde{\phi}_1}+p_2\chi_{\tilde{\phi}_2}).
		\end{eqnarray*}
		When the dimension $|A|$ takes other values, we can similarly construct many invertible real superchannels such that
		\begin{eqnarray*}
			\Theta_1(\cdots\Theta_k(p_1\chi_{\phi_1}\oplus p_2\chi_{\phi_2})\cdots)=p_1\chi_{\tilde{\phi}_1}+p_2\chi_{\tilde{\phi}_2}.
		\end{eqnarray*}
		In this way, we can conclude the equation (\ref{3}) holds.
	\end{proof}
	
	\section{Imaginarity measures of quantum channels}
	In this section, we introduce three imaginarity measures of quantum channels, that is, measures based on the robustness, the trace norm, and the sandwiched R\'enyi relative entropy, respectively.
	\subsection{Measure based on the robustness}
	\par 
	The Measure based on the robustness has appeared in many resource theories \cite{Hickey2018,Jin2021}. Similarly, we define the imaginarity measure of quantum channels as follows.
	\begin{eqnarray*}
		C_r(\phi)=\min_{\tilde{\phi}\in\mathcal{RC}_{AB}}\left\lbrace s\geqslant 0: \frac{\chi_{\phi}+s\chi_{\tilde{\phi}}}{1+s}\triangleq\chi_{\psi}, \psi\in\mathcal{RC}_{AB}\right\rbrace 
	\end{eqnarray*}
	\par 
	We have the following result.
	\begin{theorem}\label{Th4}
		$C_r(\phi)$ is an imaginarity measure of quantum channels.
	\end{theorem}
	\par 
	Before our proof, we need to introduce pseudomixture \cite{Sanpera1998,Napoli2016,Takagi2019,Xue2021} of $\chi_{\phi}$, that is, 
	\begin{eqnarray*}
		\chi_{\phi}=[1+C_r(\chi_{\phi})]\chi_{\phi_1}-C_r(\chi_{\phi})\chi_{\phi_2},
	\end{eqnarray*}
	where $\phi_1\in\mathcal{RC}_{AB}$, and $\phi_2$ is the optimal channel of the $\min$ in the definition of $C_r(\phi)$.
	\begin{proof}
		We only verify $\bm{(C_1)}$+$\bm{(C_3)}$+$\bm{(C_4)}$ since $\bm{(C_2)}$ can be derived by $\bm{(C_3)}$ and$\bm{(C_4)}$.
		\par\noindent
		$\bm{(C_1)}$: 
		\par 
		It can be observed that $C_r(\phi)\geqslant 0$. When $C_r(\phi)=0$, this implies $\frac{\chi_{\phi}+0}{1+0}=\chi_{\psi}$, which means $\phi=\psi\in\mathcal{RC}_{AB}$. The converse is also clearly true.
		\par\noindent
		$\bm{(C_3)}$: 
		\par 
		Let $\Theta\in\mathcal{RSC}_{ABA'B'}$ with $\Theta(\cdot)=\sum_m M_m\cdot M_m^\dagger$. For any quantum channel $\phi$ with its pseudomixture, under the requirement of $\bm{(C_3)}$, we have $p_m=\text{Tr}(M_m\chi_{\phi}M_m^{\dagger})$ and $\chi_{\phi_m}=M_m\chi_{\phi}M_m^{\dagger}/p_m$. Then, we obtain:
		\begin{eqnarray*}
			&&M_m\chi_{\phi}M_m^{\dagger}
			\\
			&=&[1+C_r(\chi_{\phi})]M_m\chi_{\phi_1}M_m^{\dagger}-C_r(\chi_{\phi})M_m\chi_{\phi_2}M_m^{\dagger}.
		\end{eqnarray*}
		\par 
		Let 
		\begin{eqnarray*}
			\begin{cases}
				\chi_{\phi_1}^m=\frac{1}{1+s_m}\frac{1}{p_m}[1+C_r(\chi_{\phi})]M_m\chi_{\phi_1}M_m^{\dagger},
				\\
				\chi_{\phi_2}^m=\frac{1}{s_m}\frac{1}{p_m}C_r(\chi_{\phi})M_m\chi_{\phi_2}M_m^{\dagger},
			\end{cases}
		\end{eqnarray*}
		where $s_m=\frac{1}{p_m}C(\chi_{\phi})\text{Tr}(M_m\chi_{\phi_1}M_m^{\dagger})$.
		\par 
		We obtain $\chi_{\phi_m}=(1+s_m)\chi_{\phi_1}^m-s_m\chi_{\phi_2}^m$. This means $\frac{\chi_{\phi_m}+s_m\chi_{\phi_2}^m}{1+s_m}=\chi_{\phi_1}^m$. Thus we get $C_r(\chi_{\phi_m})\leqslant s_m$ from definition, that is, $p_mC_r(\chi_{\phi_m})\leqslant C(\chi_{\phi})\text{Tr}(M_m\chi_{\phi_1}M_m^{\dagger})$.
		\par 
		Therefore, we have:
		\begin{eqnarray*}
			&&\sum_mp_mC_r(\chi_{\phi_m})
			\\
			&\leqslant&\sum_mC(\chi_{\phi})\text{Tr}(M_m\chi_{\phi_1}M_m^{\dagger})
			\\
			&=&C(\chi_{\phi})\sum_m\text{Tr}(M_m^{\dagger}M_m\chi_{\phi_1})
			\\
			&=&C(\chi_{\phi})\text{Tr}(\sum_mM_m^{\dagger}M_m\chi_{\phi_1})
			\\
			&=&C(\chi_{\phi})
		\end{eqnarray*}
		This means the condition $\bm{(C_3)}$ holds.
		\par\noindent
		$\bm{(C_4)}$: 
		\par 
		For simplicity, we consider only two quantum channels $\phi_1$ and $\phi_2$, and define $\phi=p\phi_1+(1-p)\phi_2$ with any $p\in (0, 1)$. This leads to $\chi_{\phi}=p\chi_{\phi_1}+(1-p)\chi_{\phi_2}$. Let $s_p=pC_r(\chi_{\phi_1})+(1-p)C_r(\chi_{\phi_2})$.
		\par 
		Consider the pseudomixtures of $\chi_{\phi_1}$ and $\chi_{\phi_2}$:
		\begin{eqnarray*}
			\chi_{\phi_1}=[1+C_r(\chi_{\phi_1})]\chi_{\phi_{1,1}}-C_r(\chi_{\phi_1})\chi_{\phi_{1,2}},
			\\ \chi_{\phi_2}=[1+C_r(\chi_{\phi_2})]\chi_{\phi_{2,1}}-C_r(\chi_{\phi_2})\chi_{\phi_{2,2}},
		\end{eqnarray*}
		then let
		\begin{eqnarray*}
			\begin{cases}
				\chi_{\psi}=\frac{p[1+C_r(\chi_{\phi_1})]\chi_{\phi_{1,1}}+(1-p)[1+C_r(\chi_{\phi_2})]\chi_{\phi_{2,1}}}{1+s_p},
				\\
				\chi_{\tilde{\phi}}=\frac{pC_r(\chi_{\phi_1})\chi_{\phi_{1,2}}+(1-p)C_r(\chi_{\phi_2})\chi_{\phi_{2,2}}}{s_p}.
			\end{cases}
		\end{eqnarray*}
		We have $\chi_{\phi}=(1+s_p)\chi_{\psi}-s_p\chi_{\tilde{\phi}}$, that is, $\chi_{\psi}=\frac{\chi_{\phi}+s_p\chi_{\tilde{\phi}}}{1+s_p}$, thus we obtain $s_p\geqslant C_r(\chi_{\phi})$.
		\par 
		Therefore, $pC_r(\chi_{\phi_1})+(1-p)C_r(\chi_{\phi_2})=s_p\geqslant C_r(\chi_{\phi})=C(p\chi_{\phi_1}+(1-p)\chi_{\phi_2})$. The condition $\bm{(C_4)}$ holds.
	\end{proof}
	
	\subsection{Measure based on the trace norm}
	\par 
	The trace norm, 
	\begin{eqnarray*}
		C_t(\phi)=||\chi_{\phi}-\chi_{\phi}^*||_{\text{Tr}}=\text{Tr}|\chi_{\phi}-\chi_{\phi}^*|.
	\end{eqnarray*}
	can be used to construct a measure, where $\chi_{\phi}^*$ means the conjugate of $\chi_{\phi}$ and $|\chi_{\phi}-\chi_{\phi}^*|=\sqrt{(\chi_{\phi}-\chi_{\phi}^*)^{\dagger}(\chi_{\phi}-\chi_{\phi}^*)}$. 
	\begin{theorem}\label{Th5}
		$C_t(\phi)$ is a proper imaginarity measure of quantum channels.
	\end{theorem}
	\begin{proof}
		The condition $\bm{(C_1)}$ holds due to the definition of the trace norm. 
		\par 
		Due to $||\chi_{\phi}-\chi_{\phi}^*||_{\text{Tr}}\geqslant ||\Theta(\chi_{\phi})-\Theta(\chi_{\phi}^*)||_{\text{Tr}}$ for any CPTP map \cite{Nielsen2000}, it certainly holds for real superchannels. Since $\Theta$ is a real superchannel, it can be seen that $[\Theta(\chi_{\phi})]^*=\Theta(\chi_{\phi}^*)$ due to 
		\begin{eqnarray*}
			&&[\Theta(\chi_{\phi})]^*
			\\
			&=&\left[ \sum_mM_m\chi_{\phi}M_n^{\dagger}\right]^*
			\\
			&=&\sum_mM_m\chi_{\phi}^*M_n^{\dagger}
			\\
			&=&\Theta(\chi_{\phi}^*).
		\end{eqnarray*}
		Thus $\bm{(C_2)}$ is satisfied. 
		\par 
		Under the requirement of $\bm{(C_5)}$, from Property \ref{Pr4}, we have:
		\begin{eqnarray*}
			&&C_t(\chi_{\phi})
			\\
			&=&C_t(p_1\chi_{\phi_1}\oplus p_2\chi_{\phi_2})
			\\
			&=&||(p_1\chi_{\phi_1}\oplus p_2\chi_{\phi_2})-(p_1\chi_{\phi_1}\oplus p_2\chi_{\phi_2})^*||_{\text{Tr}}
			\\
			&=&||(p_1\chi_{\phi_1}\oplus p_2\chi_{\phi_2})-(p_1\chi_{\phi_1}^*\oplus p_2\chi_{\phi_2}^*)||_{\text{Tr}}
			\\
			&=&||p_1(\chi_{\phi_1}-\chi_{\phi_1}^*)\oplus p_2(\chi_{\phi_2}-\chi_{\phi_2}^*)||_{\text{Tr}}
			\\
			&=&p_1||(\chi_{\phi_1}-\chi_{\phi_1}^*)||_{\text{Tr}}+p_2||(\chi_{\phi_2}-\chi_{\phi_2}^*)||_{\text{Tr}}
			\\
			&=&p_1C_t(\chi_{\phi_1})+p_2C_t(\chi_{\phi_2}).
		\end{eqnarray*}
		Therefore, $C_r(\phi)$ satisfies $\bm{(C_1)}$+$\bm{(C_2)}$+$\bm{(C_5)}$.
	\end{proof}
	
	\subsection{Measure based the sandwiched R\'enyi relative entropy}
	\par 
	The sandwiched R\'enyi relative entropy \cite{MuellerLennert2013} is defined for quantum states with the parameter $\alpha$ as follows:
	\begin{equation}
		D_{\alpha}(\rho||\sigma)=\frac{1}{\alpha-1}\text{log }\text{Tr}(\sigma^{\frac{1-\alpha}{2\alpha}}\rho\sigma^{\frac{1-\alpha}{2\alpha}})^{\alpha},
		\nonumber
	\end{equation}
	Now, we will introduce the imaginarity measure of quantum channels based on the sandwiched R\'enyi relative entropy, defined as:
	\begin{eqnarray*}
		C_{\alpha}(\phi)=1-\text{Tr}\left[(\chi_{\phi}^*)^{\frac{1-\alpha}{2\alpha}}\chi_{\phi}(\chi_{\phi}^*)^{\frac{1-\alpha}{2\alpha}}\right]^{\alpha},
	\end{eqnarray*}
	where $\frac{1}{2}\leqslant\alpha<1$.
	We have the following results:
	\begin{theorem}\label{Th6}
		$C_{\alpha}(\phi)$ is a proper imaginarity measure of quantum channels.
	\end{theorem}
	\begin{proof}
		It is shown \cite{MuellerLennert2013} that $D_{\alpha}(\rho||\sigma)\geqslant 0$ when $\frac{1}{2}\leqslant\alpha<1$, and equation holds if and only if $\rho=\sigma$, which implies that $\text{Tr}(\sigma^{\frac{1-\alpha}{2\alpha}}\rho\sigma^{\frac{1-\alpha}{2\alpha}})^{\alpha}\leqslant 1$. Thus, considering Choi state as a quantum state, we find that $C_{\alpha}(\phi)\geqslant 0$. The equation holds if and only if $\chi_{\phi}=\chi_{\phi}^*$, satisfying the condition $\bm{(C_1)}$. 
		\par 
		From \cite{MuellerLennert2013}, we have:
		\begin{eqnarray*}
			&&D_{\alpha}(\chi_{\phi}||\chi_{\phi}^*)\geqslant D_{\alpha}\left[ \Theta(\chi_{\phi})||\Theta(\chi_{\phi}^*)\right]
			\\
			&\iff&\text{log }\text{Tr}\left[ (\chi_{\phi}^*)^{\frac{1-\alpha}{2\alpha}}\chi_{\phi}(\chi_{\phi}^*)^{\frac{1-\alpha}{2\alpha}}\right]^{\alpha}
			\\
			&&\leqslant\text{log }\text{Tr}\left\lbrace  [\Theta(\chi_{\phi}^*)]^{\frac{1-\alpha}{2\alpha}}\Theta(\chi_{\phi})[\Theta(\chi_{\phi}^*)]^{\frac{1-\alpha}{2\alpha}}\right\rbrace^{\alpha}
			\\
			&\iff&\text{Tr}\left[ (\chi_{\phi}^*)^{\frac{1-\alpha}{2\alpha}}\chi_{\phi}(\chi_{\phi}^*)^{\frac{1-\alpha}{2\alpha}}\right]^{\alpha}
			\\
			&&\leqslant\text{Tr}\left\lbrace  [\Theta(\chi_{\phi}^*)]^{\frac{1-\alpha}{2\alpha}}\Theta(\chi_{\phi})[\Theta(\chi_{\phi}^*)]^{\frac{1-\alpha}{2\alpha}}\right\rbrace^{\alpha}
			\\
			&\iff&\text{Tr}\left[ (\chi_{\phi}^*)^{\frac{1-\alpha}{2\alpha}}\chi_{\phi}(\chi_{\phi}^*)^{\frac{1-\alpha}{2\alpha}}\right]^{\alpha}
			\\
			&&\leqslant\text{Tr}\left\lbrace  ([\Theta(\chi_{\phi})]^*)^{\frac{1-\alpha}{2\alpha}}\Theta(\chi_{\phi})([\Theta(\chi_{\phi})]^*)^{\frac{1-\alpha}{2\alpha}}\right\rbrace^{\alpha}
			\\
			&\iff&C_{\alpha}(\phi)\geqslant C_{\alpha}\left[\Theta(\phi)\right].
		\end{eqnarray*}
		Thus $\bm{(C_2)}$ holds.
		\par 
		Under the requirement of $\bm{(C_5)}$, we have:
		\begin{eqnarray*}
			&&\text{Tr}\left[(\chi_{\phi}^*)^{\frac{1-\alpha}{2\alpha}}\chi_{\phi}(\chi_{\phi}^*)^{\frac{1-\alpha}{2\alpha}}\right]^{\alpha}
			\\
			&=&\text{Tr}\left\lbrace \bigoplus_{j=1}^2 p_j^{\frac{1-\alpha}{2\alpha}}p_jp_j^{\frac{1-\alpha}{2\alpha}}(\chi_{\phi_j}^*)^{\frac{1-\alpha}{2\alpha}}\chi_{\phi_j}(\chi_{\phi_j}^*)^{\frac{1-\alpha}{2\alpha}}\right\rbrace^{\alpha}
			\\
			&=&\sum_{j=1}^{2}p_j\text{Tr}\left[ (\chi_{\phi_j}^*)^{\frac{1-\alpha}{2\alpha}}\chi_{\phi_j}(\chi_{\phi_j}^*)^{\frac{1-\alpha}{2\alpha}}\right]^{\alpha}
		\end{eqnarray*}
		In this way, we obtain:
		\begin{eqnarray*}
			&&C_{\alpha}(\phi)
			\\
			&=&1-\text{Tr}\left[(\chi_{\phi}^*)^{\frac{1-\alpha}{2\alpha}}\chi_{\phi}(\chi_{\phi}^*)^{\frac{1-\alpha}{2\alpha}}\right]^{\alpha}
			\\
			&=&\sum_{j=1}^{2}p_j\left\lbrace 1-\text{Tr}\left[ (\chi_{\phi_j}^*)^{\frac{1-\alpha}{2\alpha}}\chi_{\phi_j}(\chi_{\phi_j}^*)^{\frac{1-\alpha}{2\alpha}}\right]^{\alpha}\right\rbrace 
			\\
			&=&p_1C_{\alpha}(\phi_1)+p_2C_{\alpha}(\phi_2)
		\end{eqnarray*}
		The condition $\bm{(C_5)}$ holds.
	\end{proof}
	\section{Conclusion}
	\par 
	We introduce the framework of imaginarity of quantum channels, and introduce an alternative framework. Besides, we give some properties of imaginarity measure of quantum channels. We also give three imaginarity measures for quantum channels. They are $C_r$, $C_t$ and $C_{\alpha}$ based on the robustness, the trace norm and the sandwiched R\'enyi relative entropy, respectively.
	
	\section*{Acknowledgement}
	\par 
	We thank Jianwei Xu for helpful discussions and comments. This work was supported by National Natural Science Foundation of China (Grants No. 12271474).


\begin{thebibliography}{10}
		
		\bibitem{Hickey2018}
		A. Hickey and G. Gour.
		\newblock Quantifying the Imaginarity of Quantum Mechanics.
		\newblock {\em Journal of Physics A: Mathematical and Theoretical}, \textbf{51}, 414009(2018).
		
		\bibitem{Wu2021}
		K. D. Wu, T. V. Kondra, S. Rana, C. M. Scandolo, G. Y. Xiang, C. F. Li, G. C. Guo and A. Streltsov.
		\newblock Operational Resource Theory of Imaginarity.
		\newblock {\em Physical Review Letters}, \textbf{126}, 090401(2021).
		
		\bibitem{Wu2023}
		K. D. Wu, T. V. Kondra, C. M. Scandolo, S. Rana, G. Y. Xiang, C. F. Li, G. C. Guo and A. Streltsov.
		\newblock Resource Theory of Imaginarity: New Distributed Scenarios.
		\newblock {\em arXiv:2301.04782}, (2023).
		
		\bibitem{Chen2022}
		Q. Chen, T. Gao and F. L. Yan.
		\newblock Measures of imaginarity and quantum state order.
		\newblock {\em Science China Physics, Mechanics \& Astronomy}, \textbf{66}, 1(2023).
		
		\bibitem{Xu2023}
		J. W. Xu.
		\newblock Quantifying the imaginarity of quantum states via tsallis relative entropy.
		\newblock {\em arXiv:2311.12547}, (2023).
		
		\bibitem{Chen2024}
		X. Y. Chen and Q. Lei.
		\newblock Imaginarity measure induced by relative entropy.
		\newblock {\em arXiv:2404.00637}, (2024).
		
		\bibitem{Xue2021}
		S. N. Xue, J. Guo, P. Li, M. F. Ye and Y. M. Li.
		\newblock Quantification of resource theory of imaginarity.
		\newblock {\em Quantum Information Processing}, \textbf{20}, 383(2021).
		
		\bibitem{Mani2015}
		A. Mani and V. Karimipour.
		\newblock Cohering and decohering power of quantum channels.
		\newblock {\em Physical Review A}, \textbf{92}, 032331(2015).
		
		\bibitem{Dana2017}
		K. B. Dana, M. G. Díaz, M. Mejatty, and A. Winter.
		\newblock Resource theory of coherence: Beyond states.
		\newblock {\em Physical Review A}, \textbf{95}, 062327(2017).
		
		\bibitem{Leditzky2018}
		F. Leditzky, E. Kaur, N. Datta and M. M. Wilde.
		\newblock Approaches for approximate additivity of the Holevo information of quantum channels.
		\newblock {\em Physical Review A}, \textbf{97}, 012332(2018).
		
		\bibitem{Yuan2019}
		X. Yuan.
		\newblock Hypothesis testing and entropies of quantum channels.
		\newblock {\em Physical Review A}, \textbf{99}, 032317(2019).
		
		\bibitem{Wang2019}
		X. Wang, M. M. Wilde and Y Su.
		\newblock Quantifying the magic of quantum channels.
		\newblock {\em New Journal of Physical}, \textbf{21}, 103002(2019).
		
		\bibitem{Liu2019}
		Z. W. Liu and A. Winter.
		\newblock Resource theories of quantum channels and the universal role of resource erasure.
		\newblock {\em arXiv:1904.04201}, (2019).
		
		\bibitem{Xu2019}
		J. W. Xu.
		\newblock Coherence of quantum channels.
		\newblock {\em Physical Review A}, \textbf{100}, 052311(2019).
		
		\bibitem{Zhou2022}
		H. Q. Zhou, T. Gao and F. L. Yan.
		\newblock Quantifying the entanglement of quantum channel.
		\newblock {\em Physical Review A}, \textbf{4}, 013200(2022).
		
		\bibitem{Xu2021}
		J. W. Xu.
		\newblock Coherence of quantum Gaussian channels.
		\newblock {\em Physical Letters A}, \textbf{387}, 127028(2021).
		
		\bibitem{Luo2022}
		Y. Luo, M. F. Ye and Y. M. Li. 
		\newblock Coherence weight of quantum channels.
		\newblock {\em Physica A}, \textbf{599}, 127510(2022).
		
		\bibitem{Jin2021}
		Z. X. Jin, L. M. Yang, S. M. Fei, X. Li-Jost, Z. X. Wang, G. L. Long and C. F. Qiao.
		\newblock Maximum relative entropy of coherence for quantum channels.
		\newblock {\em  Science China Physics, Mechanics \& Astronomy}, \textbf{64}, 280311(2021).
		
		\bibitem{Fan2022}
		Y. J. Fan, X. Guo and X. Y. Yang.
		\newblock Quantifying coherence of quantum channels via trace distance.
		\newblock {\em  Quantum Information Processing}, \textbf{21}, 339(2022).
		
		\bibitem{Ye2024}
		M. F. Ye, Y. Luo and Y. M. Li.
		\newblock Quantifying channel coherence via the norm distance.
		\newblock {\em Journal of Physics A: Mathematical and Theoretical}, \textbf{57}, 015307(2024).
		
		\bibitem{Fan2024}
		J. R. Fan, Z. Q. Wu and S. M. Fei.
		\newblock Quantifying coherence of quantum channels based on the generalized $\alpha$--$z$--relative R\'enyi entropy.
		\newblock {\em Quantum Information Processing}, \textbf{23}, 100(2024).
		
		\bibitem{Yu2016}
		X. D. Yu, D. J. Zhang, G. F. Xu and D. M. Tong.
		\newblock Alternative framework for quantifying coherence.
		\newblock {\em Physical Review A}, \textbf{94}, 060302(2016).
		
		\bibitem{Kong2022}
		S. Y. Kong, Y. J. Wu, Q. Q. L, Z. X. Wang and S. M. Fei.
		\newblock An Alternative Framework For Quantifying Coherence Of Quantum Channels.
		\newblock {\em International Journal of Theoretical Physics}, \textbf{61}, 113(2022).
		
		\bibitem{Qiang2022}
		Q. Chen.
		\newblock The research of measures of imaginarity under imaginarity resource theory.
		\newblock M.S. thesis.
		\newblock Hebei Normal University, Hebei, China, (2022).
		
		\bibitem{Zanoni2024}
		E. Zanoni and C. M. Scandolo.
		\newblock Choi-Defined Resource Theories.
		\newblock {\em arXiv:2402.12569}, (2024).
		
		\bibitem{Chiribella2008}
		G. Chiribella, G. M. D'Ariano and P. Perinotti.
		\newblock Transforming quantum operations: Quantum supermaps.
		\newblock {\em Europhysics Letters}, \text{83}, 30004(2008).
		
		\bibitem{Sanpera1998}
		A. Sanpera, R. Tarrach and G. Vidal.
		\newblock Local description of quantum inseparability.
		\newblock {\em Physical Review A}, \textbf{58}, 826(1998).
		
		\bibitem{Napoli2016}
		C. Napoli, T. R. Bromley, M. Cianciaruso, M. Piani, N. Johnston and G. Adesso.
		\newblock Robustness of Coherence: An Operational and Observable Measure of Quantum Coherence inseparability.
		\newblock {\em Physical Review Letters}, \textbf{116}, 150502(2016).
		
		\bibitem{Takagi2019}
		R. Takagi and B. Regula.
		\newblock General Resource Theories in Quantum Mechanics and Beyond: Operational Characterization via Discrimination Tasks.
		\newblock {\em Physical Review X}, \textbf{9}, 826(2019).
		
		
		\bibitem{Nielsen2000}
		M. A. Nielsen, I. L. Chuang.
		\newblock Quantum Computation and Quantum Information.
		\newblock {\em Cambridge University Press}, (2000).
		
		
		\bibitem{MuellerLennert2013}
		M. M{\"u}ller-Lennert, F. Dupuis, O. Szehr, S. Fehr and M. Tomamichel.
		\newblock On quantum r\'enyi entropies: a new definition, some properties and several conjectures.
		\newblock {\em Journal of Mathematical Physics}, \textbf{54}, 122203(2013).
		
	\end{thebibliography}
\end{document}